%% file: RecCoindLogic.tex
\documentclass[12pt]{article}
\usepackage[a4paper, total={6in, 8in}]{geometry}

\usepackage[T1]{fontenc}
\usepackage[utf8]{inputenc}
\usepackage[british]{babel}

\usepackage{thm}
\usepackage{libertine}
\usepackage[libertine,cmintegrals,cmbraces]{newtxmath}

\usepackage{amsmath,amssymb}
\usepackage{mathtools}
\usepackage[f]{esvect}
\usepackage{hhline}
\usepackage{xparse}
\usepackage{enumitem}
\usepackage{tikz}
\usepackage{tikz-cd}
\usepackage[obeyspaces,hyphens,spaces]{url}
\usepackage{hyperref}
\usepackage{bussproofs}
\usepackage{verbatim}   % for the comment environment

\usepackage{local-macros}

\newcommand*{\colonequals}{\coloneqq}

\let\underbrace\LaTeXunderbrace

\usepackage{ifthen}
\ifthenelse{\isundefined{\development}}{
}{
  \overfullrule=1mm
}

\newcommand\blfootnote[1]{%
  \begingroup
  \renewcommand\thefootnote{}\footnote{#1}%
  \addtocounter{footnote}{-1}%
  \endgroup
}

\newenvironment{longproof}%
{%
  \ifthenelse{\isundefined{\showproofs}}%
  {\expandafter\comment}%
  {\begin{proof}[Proof Details]}%
}%
{%
  \ifthenelse{\isundefined{\showproofs}}%
  {\expandafter\endcomment}%
  {\end{proof}}%
}

\ifthenelse{\isundefined{\development}}{
}{
  \usepackage[notref,notcite]{showkeys}
  \renewcommand*{\showkeyslabelformat}[1]{%
    \fbox{\parbox[t]{1.5cm}{\raggedright\normalfont\scriptsize\url{#1}}}}
  \usepackage{gitinfo}
}

\ifthenelse{\isundefined{\development}}{
  \usepackage[disable]{todonotes}
}{
  \usepackage[textwidth=1.7cm,colorinlistoftodos]{todonotes}
}

\usetikzlibrary{calc}
\tikzstyle{process} = [rectangle, minimum width=2cm, minimum height=0.8cm,
text centered, draw=black]
\tikzstyle{arrow} = [thick,->,>=stealth]

\title{Breaking the Loop \\
  Recursive Proofs for Coinductive Predicates in Fibrations}
\author{
  Henning Basold \\
  CNRS, ENS de Lyon \\
  \texttt{henning.basold@ens-lyon.fr}
}
\date{}

\begin{document}

\maketitle{}

\begin{abstract}
  The purpose of this paper is to develop and study recursive proofs of
  coinductive predicates.
  Such recursive proofs allow one to discover proof goals in the construction
  of a proof of a coinductive predicate, while still allowing the use of
  up-to techniques.
  This approach lifts the burden to guess invariants, like bisimulation
  relations, beforehand.
  Rather, they allow one to start with the sought-after proof goal and
  develop the proof from there until a point is reached, at which the proof
  can be closed through a recursion step.
  Proofs given in this way are both easier to construct and to understand,
  similarly to proofs given in cyclic proof systems or by appealing
  parameterised coinduction.

  In this paper, we develop a framework for recursive proofs of coinductive
  predicates that are given through fibrational predicate liftings.
  This framework is built on the so-called later modality, which has made its
  appearance in type theoretic settings before.
  In particular, we show the soundness and completeness of recursive
  proofs, we prove that compatible up-to techniques can be used as
  inference rules in recursive proofs, and provide some illustrating examples.
\end{abstract}

\ifthenelse{\isundefined{\development}}{
}{
  \blfootnote{Base revision~\gitAbbrevHash, \gitReferences~from~\gitAuthorDate}
}

\input{content/intro}
\input{content/sequences}
\input{content/fib-seq}
\input{content/fin-chain-up-to}
\input{content/examples}
\input{content/well-founded}
\input{content/concl}

\bibliographystyle{splncs03}
\bibliography{RecCoindLogic}

\end{document}

%% file: content/intro.tex
\section{Introduction}
\label{sec:intro}

Recursion is one of the most fundamental notions in Computer Science and
Mathematics, be it as the foundation of computability, or to define and
reason about structures determined by repeated constructions.
In this paper, we will focus on the use of recursion as a proof method
for coinductive predicates.

The usual way to prove that some objects are contained in a coinductive
predicate or are related by a coinductive relation, is to establish
an invariant.
More specifically, suppose $\Phi \from L \to L$ is a monotone
function on a lattice and $\Phi$ that has a greatest fixed point
$\nu \Phi$.
One proves that the coinductive predicate $\nu \Phi$ holds for $x \in L$
by establishing a $y \in L$ with $x \leq y \leq \Phi(y)$.
%, from which $x \leq \nu \Phi$ follows.
This approach does, however, not fit common practice, as one usually
incrementally constructs the invariant $y$, rather than guessing it, while
following the necessary proof steps.
Such an incremental construction leads to a recursive proof methodology.

There are several ways that have been proposed to formalise the idea of
recursive proofs for coinductive predicates.
In the setting of complete lattices, Hur et al.~\cite{Hur13:ParameterizedCoind}
developed so-called parameterised coinduction.
Their techniques were later streamlined using the companion by
Pous~\cite{Pous:FinalUpTo-LICS16}.
Another approach is to use ideas from game theory~\cite{Niwinski96:GamesMuCalc,%
  Santocanale02:muBicompleteParity}
to prove coinductive predicates.
There are also type theoretic approaches that use systems of equations
to prove coinductive predicates~\cite{Abel:SizedTypes,%
  Birkedal:GuardedRecUniverseFP,Bizjak16:GuardedDTT,Gimenez-RecursiveSchemes}.
Finally, recursion has also entered syntactic proof systems in the form of
cyclic proof systems, e.g.~\cite{Brotherston05:CyclicFOLInductiveDef,%
  Cockett01:Deforestation,Dax06:ProofSysLinearTimeMuCalc,%
  RosoLucanu09:CircCoinduction,Santocanale02:CircProofs}.
Cyclic proof systems are particularly useful in settings that require proofs by
induction or coinduction because cyclic proof systems ease proofs enormously
compared to, for example, invariant-based method from above or (co)induction
schemes.
Nothing comes for free though:
In this case checking proofs becomes more difficult, as the correctness
conditions are typically global for a proof tree and not compositional.
For the same reason, also soundness proofs a often rather complex.

In this paper, we will study an approach to proving coinductive predicates
through recursive proofs.
Recursion in such proofs is thereby controlled by using the so-called later
modality~\cite{Nakano00:ModalityRec}, which allows checking of recursive proofs
on a per-rule basis.
This results in straightforward proof checking, a per-rule soundness proof,
and proofs that can be easily debugged.
We will thereby develop the recursive proofs abstractly for a general
first-order logic, given in form of a fibration.
This generality allows us to obtain recursive proofs for coinductive predicates
in many different settings.
In particular, we will discuss set-based predicates, quantitative
predicates, syntactic first-order logic, and (models of) dependent type theory.
An instance of this is the syntactic first-order logic given by the author
in~\cite{Basold17:Phd} to reason about program equivalences.
This instance was also the original motivation of the present paper, as the
results in loc.~cit. are mostly obtained by hand.

Towards this, we proceed as follows.
In \iSecRef{func-cat-fib}, we show that certain fibrations of functors are
fibred Cartesian closed, which is the technical machinery that makes recursive
proofs work.
Next, we develop in \iSecRef{sequences} and \iSecRef{fib-seq} a theory of
descending chains of predicates in general categories and fibrations,
respectively.
In the same sections, we also provide the necessary results for the construction
of recursive proofs.
\secRef{fin-chain-up-to} provides some specific results concerning the
descending chain that is induced by a lifting of a behaviour functor.
In particular, we show how up-to techniques can be used as proof rules.
We instantiate these results in \iSecRef{examples} to obtain recursive proofs
for some illustrative examples.

\paragraph{Related Work}

To a large part, the present paper develops many results of
Birkedal et~al.~\cite{Birkedal12:GuardedDomainTheory} in the setting of
general fibrations rather than just the codomain fibration
$\ArrC{\SetC} \to \SetC$ of sets.
That~\cite{Birkedal12:GuardedDomainTheory} was so restrictive is not so
surprising, as the intention there was to construct models of programming
languages, rather than applying the developed techniques to proofs.
Going beyond the category of sets also means that one has to involve much
more complicated machinery to obtain exponential objects.
Later, Bizjak et al.~\cite{Bizjak16:GuardedDTT} extended the techniques
from~\cite{Birkedal12:GuardedDomainTheory} to dependent type
theory, thereby enabling reasoning by means of
recursive proofs in a syntactic type theory.
However, also this is again a very specific setting, which rules out most
examples that we are interested in here.
Similarly, also the parameterised coinduction in~\cite{Hur13:ParameterizedCoind}
is too restrictive, as it applies only to lattices.
It might be possible to develop parameterised coinduction in the setting
of fibrations by using the companion~\cite{Pous:FinalUpTo-LICS16,%
  Pous17:CompanionCodensCausal,BPR17:MonoidalCompany}.
We leave this for another time though.

%% file: content/sequences.tex
\section{Functor Categories and Fibrations}
\label{sec:func-cat-fib}

We fix an index category $\Cat{I}$ in the following and define
\begin{equation*}
  \AxiomC{$F \from \Cat{C} \to \Cat{D}$}
  \UnaryInfC{$\chLift{F} \from \ICat{\CCat} \to \ICat{D}$}
  \DisplayProof
\end{equation*}
by $\chLift{F}(\sigma) = F \comp \sigma$.
Note that $\chLift{F} = \IntHom{I}{F}$, where
$\IntHom{I}{-} \from \CatC \to \CatC$ is the strict 2-functor
that assigns to a category $\Cat{C}$ the functor category $\ICat{\CCat}$.
Thus, $\chLift{(-)}$ preserves composition of functors and applies to natural
transformations as well.
We use this to define for a morphism $f \from X \to Y$ in $\Cat{C}$, a morphism
$\chLift{f} \from K_X \natTo K_Y$ in $\ICat{\CCat}$ where $K_X$ is the constant
functor sending any object in $I$ to $X$:
Note that there is a natural transformation $K_f \from K_X \natTo K_Y$, which
is given by $K_{f,I} = f$.
Thus, we can put $\chLift{f} = \IntHom{\Cat{I}}{K_f}$.

\begin{lemma}
  \label{lem:indexed-lifting}
  If $F \from \Cat{C} \to \Cat{D}$ and $G \from \Cat{D} \to \Cat{C}$
  with $F \dashv G$, then $\chLift{F} \dashv \chLift{G}$.
\end{lemma}
\begin{proof}
  Given $F$ and $G$ as above, the following unique correspondence
  follows from the point-wise unique correspondence given by the adjunction
  $F \dashv G$.
  % \begin{equation*}
  %   \alwaysDoubleLine
  %   \Axiom$\chLift{F} \sigma \fCenter \; \natTo \tau$
  %   % \UnaryInf$\all{X \in \Cat{I}} F \sigma_X \fCenter \; \to \tau_X
  %   % \text{(natural)}$
  %   % \UnaryInf$\all{X \in \Cat{I}} \sigma_X \fCenter \; \to G \tau_X
  %   % \text{(natural)}$
  %   \UnaryInf$\sigma \fCenter \; \natTo G \tau$
  %   \DisplayProof
  % \end{equation*}
  That this correspondence is natural also follows from uniqueness
  of the point-wise correspondence.
  \qedhere
\end{proof}

\begin{lemma}
  \label{lem:seq-fibred-functor}
  The functor $\chLift{(-)}$ extends to a fibred functor on the
  (large) fibration $\FibC \to \CatC$.
\end{lemma}
\begin{proof}
  A fibration $p \from \TCat \to \BCat$ induces a fibration
  $\chLift{p} \from \chLift{\TCat} \to \chLift{\BCat}$,
  see~\cite[Ex.~1.8.8]{Jacobs1999-CLTT} and \cite{Streicher2018:FibredCats}.
  Given a map of fibration $(F, G)$, easily shows that
  $(\chLift{F}, \chLift{G})$ is again a map of fibrations.
  Finally, that $\chLift{(-)}$ is fibred follows from the fact that
  $\chLift{(-)}$ preserves strict 2-pullbacks, since it is an
  enriched right adjoint functor~\cite{Kelly82:ConceptsEnrichedCatTheory}.
  \qedhere
\end{proof}

Let $S \from \op{\Cat{I}} \times \Cat{I} \to \Cat{C}$ be a functor.
The \emph{end} of $S$ is an object $\int_{i \in I} S(i,i)$ in $\Cat{C}$
together with a universal extranatural transformation
$\pi \from \int_{i \in I} S(i,i) \to S$.
Concretely, this means that $\alpha$ is a family of morphisms indexed
by objects in $i$, such that the following diagram commutes for all
$u \from i \to j$.
\begin{equation*}
  \begin{tikzcd}
    \int_{i \in I} S(i,i) \rar{\pi_j} \dar{\pi_i}
    & S(j,j) \dar{S(u,\id)} \\
    S(i,i) \rar{S(\id, u)}
    & S(i,j)
  \end{tikzcd}
\end{equation*}
Moreover, given any other extranatural transformation $\alpha \from X \to S$
there is a unique $f \from X \to \int_{i \in I} S(i,i)$ with
$\pi_i \comp f = \alpha_i$ for every $i \in \Cat{I}$.

It is well-known that ends can be computed as certain limits in $\Cat{C}$.
By analysing carefully the necessary limits, we obtain the following result.
\begin{proposition}
  \label{prop:diagrams-CCC}
  Let $\Cat{I}$ be a small category and $\Cat{C}$ a category that has finite
  limits and for every object $i \in \Cat{I}$ products of the size of
  the coslice category $\coslice{i}{\Cat{I}}$.
  If $\Cat{C}$ is Cartesian closed, then also $\ICat{\CCat}$ is.
  The exponential object is then given by
  \begin{equation*}
    \parens*{G^F}(i) = \int_{i \to j} G(j)^{F(j)}.
  \end{equation*}
\end{proposition}
\begin{proof}
  More precisely, we define for each $i \in I$ a functor
  $S_i \from \op{(\coslice{i}{\Cat{I}})} \times \coslice{i}{\Cat{I}}
  \to \Cat{C}$
  by $S(i \to j, i \to k) = G(k)^{F(j)}$ and
  $S(f \from j' \to j, g \from k \to k') = G(g)^{F(f)}$.
  The end of $S$ is then given by the equaliser as in the following diagram.
  \begin{equation*}
    \begin{tikzcd}
      \int_{i \to j} G(i)^{F(i)} \rar{}
      & \prod_{\Hom{\Cat{I}}{i}{j}} G(j)^{F(j)}
      \arrow[shift left]{r}{}
      \arrow[shift right]{r}{}
      & \prod_{\Hom{\Cat{I}}{i}{j}, \Hom{\Cat{I}}{i}{k}} G(k)^{F(j)}
    \end{tikzcd}
  \end{equation*}
  That such an equaliser gives indeed the end of $S$ is standard.
  Note that both products range only over objects in the coslice
  category $\coslice{i}{\Cat{I}}$, hence the products exist in
  $\Cat{C}$.
  Finally, that the given definition of $G^F$ is an exponential
  object is folklore, see~\cite{Shulman:FuncCatCCC} and
  cf.~\cite[Thm.~2.12]{Street10:ComprehensiveTorsors}.
  \qedhere
\end{proof}

Given that we can construct exponential objects as certain ends, one reasonably
might expect that this also works for fibred Cartesian closed categories,
which are fibrations $p \from \TCat \to \BCat$ in which every every fibre is
Cartesian closed and reindexing preserves this structure,
see~\cite[Def. 1.8.2]{Jacobs1999-CLTT}.
To prove this, we require a suitable adaption of the co-Yoneda lemma to the
setting of fibrations.
\begin{lemma}[Fibred co-Yoneda]
  \label{lem:fibred-co-yoneda}
  Let $p \from \TCat \to \BCat$ be a cloven fibration, and suppose
  $H \from \op{\Cat{I}} \to \TCat$ and $U \from \op{\Cat{I}} \to \BCat$
  are functors, such that $p \comp H = U$.
  Then
  \begin{equation*}
    H \cong
    \int^{i \in \Cat{I}} \sum_{v \in \Hom{\Cat{I}}{-}{i}} \coprodN_{U(v)} H(i).
  \end{equation*}
\end{lemma}

\begin{theorem}
  \label{thm:fibred-diagrams-CCC}
  Let $\Cat{I}$ be a small category and $p \from \TCat \to \BCat$ a cloven
  fibration that has fibred finite limits, fibred exponents and for every object
  $i \in \Cat{I}$ fibred products of the size of
  the coslice category $\coslice{i}{\Cat{I}}$.
  Under these conditions, $\ICat{p} \from \ICat{\TCat} \to \ICat{\BCat}$
  is again a fibred CCC.
  The exponential object of $F, G \in \ICat{\TCat}_{U}$ is given by
  \begin{equation*}
    \parens*{G^F}(i) =
    \int_{v \from i \to j} \bparens{\reidx{U(v)}G(j)}^{\reidx{U(v)}F(j)}.
  \end{equation*}
\end{theorem}
\begin{proof}
  The size of the involved limits to compute the end are given in the same
  way as in \iPropRef{diagrams-CCC}.
  Note that the end is equivalently given by an end followed by a product:
  \begin{equation*}
    \int_{v \from i \to j} \bparens{\reidx{U(v)}G(j)}^{\reidx{U(v)}F(j)}
    \cong \int_{j \in \Cat{I}} \prod_{v \from i \to j}
    \bparens{\reidx{U(v)}G(j)}^{\reidx{U(v)}F(j)}
  \end{equation*}
  To show that the given exponential is right-adjoint to the product of
  functors, we consider for $H \in \ICat{\TCat}_U$ the following
  chain of natural isomorphisms.
  \begin{align*}
    \Hom{\ICat{\TCat}_U}{H}{G^F}
    & \cong \int_{i \in \Cat{I}} \Hom{\TCat_{U(i)}}{H(i)}{\parens*{G^F}(i)} \\
    % & = \int_{i \in \Cat{I}}
    % \Hom{\TCat_{U(i)}}{
    %   H(i)
    % }{
    %   \int_{v \from i \to j} \bparens{\reidx{U(v)}G(j)}^{\reidx{U(v)}F(j)}
    % } \\
    & \cong \int_{i \in \Cat{I}}
    \Hom{\TCat_{U(i)}}{
      H(i)
    }{
      \int_{j \in \Cat{I}} \prod_{v \from i \to j}
      \bparens{\reidx{U(v)}G(j)}^{\reidx{U(v)}F(j)}
    } \\
    & \cong \int_{i \in \Cat{I}} \int_{j \in \Cat{I}} \prod_{v \from i \to j}
    \Hom{\TCat_{U(i)}}{
      H(i)
    }{
      \bparens{\reidx{U(v)}G(j)}^{\reidx{U(v)}F(j)}
    } \\
    & \cong \int_{i \in \Cat{I}} \int_{j \in \Cat{I}} \prod_{v \from i \to j}
    \Hom{\TCat_{U(i)}}{
      H(i) \times \parens*{\reidx{U(v)}F(j)}
    }{
      \reidx{U(v)}G(j)
    }
    \displaybreak[0] \\
    & \cong \int_{i \in \Cat{I}} \int_{j \in \Cat{I}} \prod_{v \from i \to j}
    \Hom{\TCat_{U(i)}}{
      \coprodN_{U(v)} \parens*{H(i) \times \parens*{\reidx{U(v)}F(j)}}
    }{
      G(j)
    } \\
    & \cong \int_{i \in \Cat{I}} \int_{j \in \Cat{I}} \prod_{v \from i \to j}
    \Hom{\TCat_{U(i)}}{
      \parens*{\coprodN_{U(v)} H(i)} \times F(j)
    }{
      G(j)
    }
    \tag{*} \label{frobenius}
    \displaybreak[0] \\
    & \cong \int_{i \in \Cat{I}} \int_{j \in \Cat{I}} \prod_{v \from i \to j}
    \Hom{\TCat_{U(i)}}{
      \coprodN_{U(v)} H(i)
    }{
      G(j)^{F(j)}
    }
    \displaybreak[0] \\
    & \cong \int_{j \in \Cat{I}}
    \Hom{\TCat_{U(i)}}{
      \int^{i \in \Cat{I}} \sum_{v \from i \to j} \coprodN_{U(v)} H(i)
    }{
      G(j)^{F(j)}
    } \\
    & \cong \int_{j \in \Cat{I}}
    \Hom{\TCat_{U(i)}}{
      H(j)
    }{
      G(j)^{F(j)}
    }
    \tag{**} \label{step:fibred-co-yoneda}
    \displaybreak[0] \\
    & \cong \int_{j \in \Cat{I}}
    \Hom{\TCat_{U(i)}}{
      H(j) \times F(j)
    }{
      G(j)
    }
    \displaybreak[0] \\
    & \cong
    \Hom{\ICat{\TCat}_{U}}{
      H \times F
    }{
      G
    }
  \end{align*}
  Note that coproducts in $p$ fulfil the Frobenius property in the
  step~(\ref{frobenius}) because $p$ is a fibred CCC,
  see~\cite[Lem. 1.9.11]{Jacobs1999-CLTT}.
  Moreover, we do not need to assume the existence of coproducts along morphisms
  of $\BCat$ or further colimits explicitly, since $H(j)$ is
  isomorphic to $\int^{i \in \Cat{I}} \coprodN_{v \from i \to j} \coprodN_{U(v)} H(i)$
  by the fibred co-Yoneda lemma that we used in the
  step~(\ref{step:fibred-co-yoneda}).
  \qedhere
\end{proof}

\section{Descending Chains in Categories}
\label{sec:sequences}

In this section, we extend the development
in~\cite{Birkedal12:GuardedDomainTheory} to more general categories.
Besides giving us some intuition for the later modality, we also obtain results
that we can reuse in later sections of this paper.

Let $\omega$ be the poset of finite ordinals, i.e., $\omega = \set{0,1, \dotsc}$
with their usual order.
Since $\omega$ can be seen as a category, we can use its dual category
$\op{\omega}$ as index category, thereby obtaining a functor
$\IntHom{\op{\omega}}{-} \from \CatC \to \CatC$ as in the last section.
We will denote this functor in the following by
\begin{equation}
  \chLift{(-)} = \IntHom{\op{\omega}}{-}.
\end{equation}
The \emph{category of descending chains in $\Cat{C}$} is then the presheaf
category $\CC$, the objects of which we denote by $\sigma, \tau, \dotsc$
More explicitly, $\sigma \in \CC$ assigns as a functor
$\sigma \from \op{\omega} \to \CC$ to each $n \in \N$ an object
$\sigma_n \in \Cat{C}$ and to each pair of natural numbers with $m \leq n$
a morphism $\sigma(m \leq n) \from \sigma_n \to \sigma_m$ in $\Cat{C}$.

\begin{assumption}
Throughout this section, we assume that $\Cat{C}$ is a category with
a terminal object $\T$, finite limits and is Cartesian closed.
\end{assumption}
In particular, we get by \iPropRef{diagrams-CCC} that $\CC$ is also
Cartesian closed as follows.
Let $\finOrd{n}$ be the poset of all numbers less or equal to $n$.
Observe now for $n \in \N$ that
$\coslice{n}{\op{\omega}} = \op{(\slice{\omega}{n})} = \op{\finOrd{n}}$.
Hence, $\coslice{n}{\op{\omega}}$ is finite and, as assumed, we only need
finite limits in $\Cat{C}$ to obtain Cartesian-closure of $\CC$ from
\iPropRef{diagrams-CCC}.

Let us now introduce the later modality, which is the central construction
that underlies the recursive proofs that we develop in this paper.
\begin{definition}
  \label{def:later}
  The \emph{later modality} on $\CC$ is the functor $\later \from \CC \to \CC$
  given on objects by
  \begin{equation*}
    \begin{aligned}
      (\later \sigma)_0 & = \T
      \\
      (\later \sigma)_{n+1} & = \sigma_n
    \end{aligned}
    \qquad
    (\later \sigma)(m \leq n) =
    \begin{cases}
      ! \from \sigma_n \to \T, & m = 0 \text{ or } n = 0 \\
      \sigma(m' \leq n'), & m = m'+1, n = n'+1
    \end{cases}
  \end{equation*}
\end{definition}

\begin{theorem}
  \label{thm:later-next}
  The map $\later$ given in \iDefRef{later} on objects is a functor
  $\CC \to \CC$.
  Moreover, $\later$ has a left adjoint and thereby preserves limits.
  Finally, there is a natural transformation
  \begin{equation*}
    \nextOp \from \Id \natTo \later,
  \end{equation*}
  given by $\nextOp_{\sigma,0} = \, !_{\sigma_0} \from \sigma_0 \to \T$ and
  $\, \nextOp_{\sigma,n+1} = \sigma(n \leq n+1)$.
\end{theorem}
\begin{proof}
  Functoriality is given by uniqueness of maps into the final object $\T$.
  The left adjoint to $\later$ is given by $\prev$ with
  $\prev(\sigma)_n = \sigma_{n+1}$.
  \begin{equation*}
    \alwaysDoubleLine
    \AxiomC{$\prev \sigma \to \tau$}
    \UnaryInfC{$\all{n} \sigma_{n+1} \to \tau_n$}
    \UnaryInfC{$\sigma_0 \to \T \text{ and } \all{n} \sigma_{n+1} \to \tau_n$}
    \UnaryInfC{$\sigma \to \later \tau$}
    \DisplayProof
  \end{equation*}
  Finally, naturality of $\nextOp$ is given again by uniqueness of maps
  into final objects and by functoriality of chains.
  \qedhere
\end{proof}

Since $\later$ preserves in particular binary products, we obtain the following.
\begin{lemma}
  \label{lem:later-exp-distr}
  For all $\sigma, \tau \in \CC$ there is a morphism
  $\later(\sigma^\tau) \to \later \sigma^{\later \tau}$.
\end{lemma}

One the central properties of the later modality is that it allows us to
construct fixed points of certain maps in $\CC$, which are called
contractive.
\begin{definition}
  A map $f \from \tau \times \sigma \to \sigma$ in $\CC$ is called
  $g$-\emph{contractive} if $g$ is a map
  $g \from \tau \times \later \sigma \to \sigma$ with
  $f = g \comp (\id \times \nextOp_{\sigma})$.
  We call $s \from \tau \to \sigma$ a \emph{fixed point} or
  \emph{solution} for $f$, if the following diagram commutes.
  \begin{equation*}
    \begin{tikzcd}[column sep=large]
      \tau \rar{s}
      \dar[swap]{\pair{\id, s}}
      & \sigma \\
      \tau \times \sigma
      \urar[swap]{f}
    \end{tikzcd}
  \end{equation*}

  %% Explicit unfolding
  % \begin{equation*}
  %   \begin{tikzcd}[column sep=large]
  %     \tau \rar{s}
  %     \dar[swap]{\pair{\id, s}}
  %     & \sigma \\
  %     \tau \times \sigma \rar{\id \times \nextOp}
  %     & \tau \times \later \sigma \uar[swap]{f}
  %   \end{tikzcd}
  % \end{equation*}
\end{definition}

We can now show that there is a generic operator in $\CC$ that allows
us to construct fixed points.
\begin{theorem}
  \label{thm:lob}
  For every $\sigma \in \CC$ there is a unique morphism, dinatural in $\sigma$,
  \begin{equation*}
    \lob_\sigma \from \sigma^{\later \sigma} \to \sigma,
  \end{equation*}
  such that for all $g$-contractive maps $f$ the
  map $\lob_\sigma \comp \abstr{g}$ is a solution for $f$.
  Dinaturality means thereby that for all $h \from \sigma \to \tau$ the
  diagram below commutes.
  \begin{equation*}
    \begin{tikzcd}[row sep = tiny]
      & \tau^{\later \tau} \rar{\lob_\tau}
      & \tau \\
      \sigma^{\later \tau}
      \arrow{ur}{h^{\id}} \arrow{dr}[swap]{\id^{\later h}} \\
      & \sigma^{\later \sigma} \rar{\lob_\sigma}
      & \sigma \arrow{uu}[swap]{h}
    \end{tikzcd}
  \end{equation*}
\end{theorem}
\begin{proof}
  We define $\lob_n \from (\sigma^{\later \sigma})_n \to \sigma_n$ by iteration
  on $n$.
  For $0$, we put
  \begin{equation*}
    \lob_0 \colonequals
    (\sigma^{\later \sigma})_0
    \xrightarrow{\pair{\id, !}} (\sigma^{\later \sigma})_0 \times \T
    = (\sigma^{\later \sigma})_0 \times (\later \sigma)_0
    \xrightarrow{\ev_0} \sigma_0,
  \end{equation*}
  where $\ev$ is the counit of
  $(-) \times \later \sigma \dashv (-)^{\later \sigma}$.
  In the iteration step, we define
  \begin{equation*}
    \begin{tikzcd}[%
      ,row sep = 0ex
      ,/tikz/column 1/.append style={anchor=base east}
      ,/tikz/column 2/.append style={anchor=base west}
      ,column sep=3cm
      ]%
      \lob_{n+1} \colonequals
      (\sigma^{\later \sigma})_{n+1}
      \rar{
        \pair*{\id,\sigma^{\later \sigma}(n \leq n+1)}}
      & (\sigma^{\later \sigma})_{n+1} \times (\sigma^{\later \sigma})_{n} \\
      \phantom{(\sigma^{\later \sigma})_{n+1}} \rar{\id \times \lob_n}
      & (\sigma^{\later \sigma})_{n+1} \times \sigma_n \\
      \phantom{(\sigma^{\later \sigma})_{n+1}} \rar{\ev_{n+1}}
      & \sigma_{n+1},
    \end{tikzcd}
  \end{equation*}
  where $\sigma^{\later \sigma}(n \leq n+1)$ is the functorial action of
  $\sigma^{\later \sigma}$ (\iPropRef{diagrams-CCC}).
  To show that $\lob$ is the unique map making $\lob_\sigma \comp \abstr{g}$
  as solution one first shows that $\lob$ is uniquely fulfilling the equation
  $\lob_n = (\ev \comp (\id \times \nextOp) \comp \pair{\id, \lob})_n$
  by induction on $n$ and doing a small diagram chase.
  Uniqueness of solutions is then given the properties
  of the adjunction $(-) \times \later \sigma \dashv (-)^{\later \sigma}$.
  \qedhere

  % Solution diagram
  % \begin{equation*}
  %   \begin{tikzcd}[column sep=large]
  %     \tau \rar{\lob_\sigma \comp \abstr{f}}
  %     \dar[swap]{\pair{\id, \lob_\sigma \comp \abstr{f}}}
  %     & \sigma \\
  %     \tau \times \sigma \rar{\id \times \nextOp}
  %     & \tau \times \later \sigma \uar[swap]{f}
  %   \end{tikzcd}
  % \end{equation*}

  %% Alternative definition of löb
  % \begin{equation*}
  %   \lob_0 \colonequals
  %   (\sigma^{\later \sigma})_0
  %   \xrightarrow{\pi^0_0} \sigma_0^{(\later \sigma)_0}
  %   = \sigma_0^\T
  %   \to \sigma_0,
  % \end{equation*}
  % \begin{equation*}
  %   \begin{tikzcd}[%
  %     ,row sep = 0ex
  %     ,/tikz/column 1/.append style={anchor=base east}
  %     ,/tikz/column 2/.append style={anchor=base west}
  %     ,column sep=3cm
  %     ]%
  %     \lob_{n+1} \colonequals
  %     (\sigma^{\later \sigma})_{n+1}
  %     \rar{
  %       \pair*{\alpha^{n+1}_{n+1},\sigma^{\later \sigma}(n \leq n+1)}}
  %     & \sigma_{n+1}^{\sigma_n} \times (\sigma^{\later \sigma})_{n} \\
  %     \phantom{(\sigma^{\later \sigma})_{n+1}} \rar{\id \times \lob_n}
  %     & \sigma_{n+1}^{\sigma_n} \times \sigma_n \\
  %     \phantom{(\sigma^{\later \sigma})_{n+1}} \rar{\ev}
  %     & \sigma_{n+1},
  %   \end{tikzcd}
  % \end{equation*}
  % where $\alpha^{n+1}$ is the projection of the end
  % $\int_{m \leq n+1} \sigma_m^{\tau_m}$ and
  % $\sigma^{\later \sigma}(n \leq n+1)$
  % the functorial action of $\sigma^{\later \sigma}$ (\iThmRef{chains-CCC}).
\end{proof}

\begin{remark}
  Birkedal et al.~\cite{Birkedal12:GuardedDomainTheory} give some closure
  properties of contractive maps.
  These can be extended to our more general setting, but as we will not need
  them here, we will not state and prove them.
  \qedDef
\end{remark}

%% file: content/fib-seq.tex
\section{Descending Chains in Fibrations}
\label{sec:fib-seq}

Now that we have developed some understanding of how descending chains work
in general categories, we will essentially lift the results from
\iSecRef{sequences} to fibrations.
This will allow us to construct from a first-order logic, given by a fibration,
a new logic of descending chains that admits the same logical structure as
the given fibration and admits recursive proofs for coinductive predicates.

Throughout this section, we assume the following.
\begin{assumption}
  Let $p \from \TCat \to \BCat$ be a cloven fibration, such that,
  \begin{itemize}
  \item $\TCat$ has fibred final objects,
  \item fibred finite limits in $\TCat$ exist, %finite limits in $\BCat$ and
    and
  \item $\TCat$ is a fibred CCC. % and $\BCat$ is a CCC
  \end{itemize}
\end{assumption}

Similarly to \iSecRef{sequences}, we obtain by \iLemRef{seq-fibred-functor}
that the functor $\pC \from \TCCat \to \BCCat$ given by post-composition
is a fibration.
By the above assumptions, we then get by \iThmRef{fibred-diagrams-CCC}
that $\pC$ is a fibred CCC.
We obtain another fibred CCC by change-of-base along the diagonal functor
$\delta \from \BCat \to \BCCat$ that sends an object $I \in \BCat$ to the
constant chain $K_I \from \op{\omega} \to \BCat$,
see~\cite[Ex.~1.8.8]{Jacobs1999-CLTT} and \cite{Streicher2018:FibredCats}:
% Called $p^{\op{\omega}}$ by Jacobs and $p^{(\op{\omega})}$ by Streicher.
\begin{equation*}
  \begin{tikzcd}
    \BCat \times_{\BCCat} \TCCat
    \rar{} \dar[swap]{q} \pullback
    & \TCCat \dar{\pC} \\
    \BCat \rar{\delta}
    & \BCCat
  \end{tikzcd}
\end{equation*}
Note that for $I \in \BCCat$, the fibre of $q$ above $I$ is isomorphic
to $\TCCat_{K_I}$.
Hence, we will simplify notation in the following and just refer to
$\TCCat_{K_I}$ as $\TCCat_{I}$.
Furthermore, we note the following result, which might seem trivial at first,
but it allows us to apply, for instance, \iLemRef{indexed-lifting} to functors
between fibres of a given fibration.
\begin{lemma}
  \label{lem:fibres-are-chain-cats}
  $\TCCat_{K_I} = \chCat{\TCat_I}$, which we will denote by $\TCCat_I$.
\end{lemma}

Having worked only abstractly so far, it is about time that we give a few
examples.
There are four kinds of examples that we shall use here to illustrate different
aspects of the theory: predicates over sets, quantitative predicates,
syntactic first-order logic, and set families that model dependent types.
We begin with the simplest example, namely that of predicates.
Despite its simplicity, it is already a quite useful because it allows us
to reason about predicates and relations for arbitrary coalgebras in $\SetC$.
\begin{example}[Predicates]
  \label{ex:pred-CCC}
  A standard fibration is the fibration $\Pred \to \SetC$ of predicates,
  where an object in $\Pred$ is a predicate $(P \subseteq X)$ over a set
  $X$.
  Each fibre $\Pred_X$ has a final object $\T_X = (X \subseteq X)$ and
  the fibred binary products are given by intersection.
  Moreover, exponents also exist in $\Pred_X$ by defining
  \begin{equation*}
    Q^P = \setDef{x \in X}{x \in P \implies x \in Q}.
  \end{equation*}
  The fibration $\chCat{\Pred}$ consists then of descending chains of
  predicates.
  In particular, if $\sigma \in \chCat{\Pred}_X$, then $\sigma$ is
  a chain with $\sigma_0 \supseteq \sigma_1 \supseteq \dotsm$.
  Note now that each fibre $\Pred_X$ is a poset, hence equalisers are trivial
  and (finite) limits are just given as (finite) products.
  Hence, \iThmRef{fibred-diagrams-CCC} applies and we obtain that
  $\chCat{\Pred}$ is a fibred CCC.
  Since equalisers are trivial, it is easy to see that the exponential for
  $\sigma, \tau \in \chCat{\Pred}_X$ can be defined as follows.
  \begin{equation*}
    \parens*{\tau^{\sigma}}_n =
    \bigcap\nolimits_{m \leq n} \tau_n^{\sigma_n} \subseteq X
    % \bigcap\nolimits_{m \leq n} c(m \leq n)^{-1}\parens*{\tau_n^{\sigma_n}}
  \end{equation*}
  We end this example by noting that fibred constructions, like the above
  products and exponents, are preserved by a change-of-base,
  see~\cite[Lem. 1.8.4]{Jacobs1999-CLTT}.
  This induces thus exponents in the fibration of (binary) relations
  $\Rel \to \SetC$ and the associated fibration
  $\chCat{\Rel} \to \chCat{\SetC}$.
  Hence, one can also apply the results in this paper to reason, for example,
  about bisimilarity in coalgebras.
  \qedDef
\end{example}

Often, one is not just interested in merely logical predicates, but rather
wants to analyse quantitative aspects of system.
This is, for instance, particularly relevant for probabilistic or weighted
automata.
The following example extends the predicate fibration from \iExRef{pred-CCC}
to quantitative predicates, which gives a convenient setting to reason
about quantitative properties.
\begin{example}[Quantitative Predicates]
  We define the category of quantitative predicates $\qPred$ as follows.
  \begin{equation*}
    \qPred =
    \CatDescr{
      \text{pairs } (X, \delta)
      \text{ with } X \in \SetC \text{ and } \delta \from X \to \I
    }{
      f \from (X, \delta) \to (Y, \gamma)
      \text{ if } f \from X \to Y \text{ in } \SetC
      \text{ and } \delta \leq \gamma \comp f
    }
  \end{equation*}
  It is easy to show that the first projection $\qPred \to \SetC$ gives rise
  to a cloven fibration, for which the reindexing functors are given for
  $u \from X \to Y$ by
  \begin{equation*}
    \reidx{u}(Y, \gamma) = \bparens{X, \lam \gamma(u(x))}.
  \end{equation*}
  For brevity, let us refer to an object $(X, \delta)$ in $\qPred_X$ just by its
  underlying valuation $\delta$.
  One readily checks that $\qPred$ is a fibred CCC by defining the products
  and exponents by
  \begin{equation*}
    (\delta \times \gamma)(x) = \min \set{\delta(x), \gamma(x)}
    \quad \text{ and } \quad
    (\delta \Rightarrow \gamma)(x) =
    \begin{cases}
      1, & \delta(x) \leq \gamma(x) \\
      \gamma(x), & \text{otherwise}
    \end{cases}.
  \end{equation*}
  Fibred final objects are given by the constantly $1$ valuation.
  Again, each fibre $\qPred_X$ is a poset, hence finitely complete and
  so $\chCat{\qPred}$ is a fibred CCC.
  \qedDef
\end{example}

The original motivation for the work presented in this paper was to abstract
away from the details that are involved in constructing a syntactic logic
for a certain coinductive relation in~\cite{Basold17:Phd}.
In~\cite{Basold17:Phd}, the author developed a first-order logic that features
the later modality to reason about program equivalences.
This logic was given in a very pedestrian way, since the syntax, proof
system, model and proof system was constructed from scratch.
The proofs often involved then something along the lines of
``true because this is an index-wise interpretation of intuitionistic logic''.
Thus, the aim of the following example is to show that we can just take any
first-order logic $L$ and extend it to a logic $\chCat{L}$, in which formulas
are descending chains of formulas in $L$.
Crucially, the logic $\chCat{L}$ will have the later modality as a new
formula construction, and it will get new proof rules that correspond
to the morphism $\nextOp$, the functoriality of $\later$ and construction of
fixed points through the $\lob$ morphism.
We will also see below that quantifier can be lifted to formulas in $\chCat{L}$,
and that the later modality interacts well with conjunction,
implication and quantification,
cf.~\iThmRef{later-next} and~\iLemRef{later-exp-distr}.
After this long-winded motivation, let us now come to the actual example.
\begin{example}[Syntactic Logic]
  \label{ex:syntactic-logic-CCC}
  Suppose we are given a typed calculus, for example the simply typed
  $\lambda$-calculus, and a first-order logic, in which the variables range
  over the types of the calculus.
  More precisely, let $\Gamma$ be a context with
  $\Gamma = x_1 : A_1, \dotsc, x_n : A_n$, where the $x_i$ are variables and
  the $A_i$ are types of the calculus.
  We write then $\Gamma \Vdash t : A$ if $t$ is a term of type $A$ in context
  $\Gamma$, $\validForm{\varphi}$ if $\varphi$ is formula with variables
  in $\Gamma$, and $\Gamma \vdash \varphi$ if $\varphi$ is provable in the
  given logic.
  Let us assume that the logic also features a truth formula $\top$,
  conjunction $\conj$ and implication $\to$, which are subject to the usual
  proof rules of intuitionistic logic.
  This allows us to form a fibration as follows.
  First, we define $\mathcal{C}$ to be the category that has context $\Gamma$
  as objects and tuples $t$ of terms as morphisms $\Delta \to \Gamma$
  with $\Delta \Vdash t_i : A_i$.
  Next, we let $L$ be the category that has pairs $\sPair{\Gamma, \varphi}$
  with $\validForm{\varphi}$ as objects, and a morphism
  $(\Delta, \psi) \to (\Gamma, \varphi)$ in $L$ is given by a morphism
  $t \from \Delta \to \Gamma$ in $\mathcal{C}$ if
  $\Delta \vdash \psi \to \varphi[t]$, where $\varphi[t]$ denotes the
  substitution of $t$ in the formula $\varphi$.
  The functor $p \from L \to \mathcal{C}$ that maps $(\Gamma, \varphi)$ to
  $\Gamma$ is then easily seen to be a cloven (even split) fibration,
  see for example~\cite{Jacobs1999-CLTT}.

  We note that $p$ has fibred finite products and exponents, as the logic
  that we started with has $\top$, conjunction and implication with the
  necessary proof rules.
  Moreover, since each fibre is a pre-ordered set, equalisers are again
  trivial.
  Hence, $\pC$ is also a fibred CCC.
  Explicitly, for chains $\varphi, \psi$ of formulas in $\pC_A$ above the
  constant chain $K_A$ for a type $A$, the exponent
  $\psi \Rightarrow \varphi$ in $\pC$ is given by
  \begin{equation*}
    (\psi \Rightarrow \varphi)_n =
    \bigwedge_{m \leq n} \psi_m \to \varphi_n,
  \end{equation*}
  where $\bigwedge$ is a shorthand for a finite number of conjunctions.
  \qedDef
\end{example}

We lift now the constructions from the last \iSecRef{sequences} to the
fibres of $\TCCat$.
\begin{theorem}
  \label{thm:fibred-later}
  For each $c \in \BCCat$, there is a fibred functor
  $\laterFib{c} \from \TCCat_c \to \TCCat_c$
  given by
  \begin{align*}
    (\laterFib{c} \sigma)_0 & = \T_{c_0} \\
    (\laterFib{c} \sigma)_{n+1} & = \reidx{c(n \leq n+1)}(\sigma_n).
  \end{align*}
  Moreover, $\laterFib{c}$ preserves fibred finite products and if $p$ is a
  bifibration then $\laterFib{c}$ preserves all fibred limits.
  Finally,  there is a natural transformation
  $\nextFib{c} \from \Id \natTo \laterFib{c}$, given by
  $\nextFib{c}_{\sigma,0} = ! \from \sigma_0 \to \T_{c_0}$
  and $\nextFib{c}_{\sigma,n+1} = \sigma(n \leq n+1)$.
\end{theorem}
\begin{proof}
  We define $\later \sigma$ on morphisms by case distinction
  as follows.
  \begin{align*}
    (\later \sigma)(0 \leq 0) & = \id \from \T_{c_0} \to \T_{c_0} \\
    (\later \sigma)(0 \leq n+1) & =
    \reidx{c(n \leq n+1)}(\sigma_n)
    \xrightarrow{\cartL{c(n \leq n+1)}{\sigma_n}} \sigma_n
    \xrightarrow{!} \T_{c_0} \\
    (\later \sigma)(m+1 \leq n+1) & =
    \text{mediating morphism in the following diagram}
  \end{align*}
  \begin{equation*}
    \begin{tikzcd}
      \reidx{c(n \leq n+1)}(\sigma_n)
      \rar{} \dar[dashed]{(\later \sigma)(m+1 \leq n+1)}
      & \sigma_n \dar{\sigma(m \leq n)} \\
      \reidx{c(m \leq m+1)}(\sigma_m) \rar{}
      & \sigma_m
    \end{tikzcd}
    \xmapsto{p}
    \begin{tikzcd}
      c_{n+1} \rar{} \dar{}
      & c_n \dar{} \\
      c_{m+1} \rar{}
      & c_m
    \end{tikzcd}
  \end{equation*}
  Note that the right diagram commutes by functoriality of $c$.
  It is clear that $\pC(\later \sigma) = c$ by the above definition,
  and so $\later \sigma$ is an object in $\TCCat_c$.
  Defining $\later$ on morphisms is a straightforward, as it is to check
  functoriality.
  That $\laterFib{c}$ is preserved by reindexing, that is, for $f \from c \to d$
  in $\BCat$ one has $\reidx{f} \comp \laterFib{c} \cong \later_d \comp \reidx{f}$,
  is given by the properties of a cloven fibration.
  That $\laterFib{c}$ preserves products is a simple calculation.
  The preservation of all fibred limits if $p$ is a bifibration is given by the
  fact that $\laterFib{c}$ then has a fibred left adjoint $\prev_c$ given by
  $(\prev_c \sigma)_n = \sum_{c(n \leq n+1)} \sigma_{n+1}$.
  Finally, naturality of $\nextFib{c}$ is given as before.
  \qedhere
\end{proof}

Let us briefly stop to discuss the perspective on the later modality that
arises canonically from the development in the previous section.
\begin{remark}
  We note that we can instantiate all the results from \iSecRef{sequences} to
  $\TCat$ as follows.
  Suppose that $\BCat$ is a finitely complete CCC and $\TCat$ also has a global
  finite limits and exponents, such that the corresponding adjunctions are
  given by maps of fibrations.
  This means, for instance, that for all $X \in \TCat$ there are
  adjunctions $(-) \times X \dashv (-)^X$ and $(-) \times pX \dashv (-)^{pX}$
  on $\TCat$ and $\BCat$, respectively, such that
  $\parens*{(-) \times pX, (-) \times X}$ and
  $\parens*{(-)^{pX}, (-)^{X}}$ are maps of fibrations.
  % $(-) \times X \from \TCat \to \TCat$,
  % $(-) \times pX \from \BCat \to \BCat$,
  % $(-)^X \from \TCat \to \TCat$ and
  % $(-)^{pX} \from \TCat \to \TCat$,
  This structure gives us that $\pC \from \TCCat \to \BCCat$ has global
  exponents.
  Moreover, one can show that $(\later, \later) \from \pC \to \pC$ is a map of
  fibrations and that the next- and Löb-operations are preserved by
  $\pC$: $\pC(\nextOp) = \nextOp$ and $\pC(\lob) = \lob$.
  However, we will not make use of these results here, as their use is vastly
  more complicated than the fibred approach.
  For example, the predicate fibration has global exponents given by
  \begin{equation*}
    (P \subseteq X)^{(Q \subseteq Y)} =
    \setDef{f \from Y \to X}{\all{y \in X} f(y) \in P}
    \subseteq X^Y.
  \end{equation*}
  The problem is that we would need to show that solutions of certain morphism
  obtained through using $\lob$ are vertical, as we often want to prove
  the set inclusion of predicates.
  Since formulating and proving such conditions seem to very hard and since
  they do not even seem to be useful, we will refrain from pursuing the global
  Cartesian structure on $\pC$ further here.
  %  with
  % $\T \to \sigma^\tau$
  % $\T \to (\sigma^\tau)^{\later(\sigma^\tau)}$
  % $\later(\sigma^\tau) \to (\sigma^\tau)$
  % $\later(\sigma^\tau) \times \tau \to \sigma$
  % $\later(c^c) \times c \to c$
  % $\pC(f) = \id$.
  % Then $\pC(\lob \comp \lambda f \from \tau \to \sigma$
  \qedDef
\end{remark}

As we mentioned above, if $p$ has a global final object, then we can
instantiate \iSecRef{sequences} to the fibration $p$.
This gives us a map of fibration $(\later, \later)$ on $\pC$.
Since the fibred final objects $\T_I$ in $\TCat_I$ are related to
the final object $\T$ of $\TCat$ by $\T_I \cong \reidx{!_I}(\T)$,
we obtain that the global and local later modalities are intrinsically related
\begin{lemma}
  For all $\sigma \in \TCCat_I$, we have
  $\laterFib{c} \sigma \cong \reidx{\nextOp_c}(\later \sigma)$.
\end{lemma}
\begin{longproof}
  For each $n \in \N$, we have
  \begin{align*}
    (\laterFib{c} \sigma)_n
    & =
    \begin{cases}
      \T_{c_0}, & n = 0 \\
      \reidx{c(k \leq k+1)} \sigma_k, & n = k + 1
    \end{cases}
    \\
    & \cong
    \begin{cases}
      \reidx{!_I}(\T), & n = 0 \\
      \reidx{c(k \leq k+1)} \sigma_k, & n = k + 1
    \end{cases}
    \\
    % & =
    % \begin{cases}
    %   \reidx{!_I}(\T), & n = 0 \\
    %   \reidx{K_I(k \leq n)}(\sigma_k), & n = k + 1
    % \end{cases}
    % \\
    & = \reidx{\nextOp_{c,n}}((\later \sigma)_n)
    \\
    & = (\reidx{\nextOp_c}(\later \sigma))_n
  \end{align*}
  Thus $\laterFib{c}(\sigma) \cong \reidx{\nextOp_c}(\later \sigma)$
  as claimed.
  \qedhere
\end{longproof}

% \begin{theorem}
%   $\later$ is a fibred functor as in the following diagram
%   \begin{equation*}
%     \begin{tikzcd}
%       \TCCat \rar{\later} \dar{\pC}
%       & \TCCat \dar{\pC} \\
%       \BCCat \rar{\later}
%       & \BCCat
%     \end{tikzcd}
%   \end{equation*}
%   Moreover, the morphisms constructed in the \iSecRef{sequences}
%   are preserved by $\pC$:
%   \begin{itemize}
%   \item $\pC(\nextOp) = \nextOp$
%   \item $\pC(\lob) = \lob$ (\iThmRef{lob})
%   \item $\pC(\ch{\rho}) = \ch{\pC(\rho)}$ (\iThmRef{compat-proof-rules})
%   \end{itemize}
% \end{theorem}

Due to \iLemRef{fibres-are-chain-cats}, we can apply many construction
easily point-wise to chains with constant index.
For instance, we can lift products and coproducts in the following sense.
\begin{theorem}
  \label{thm:coproduct-lift}
  If for $f \from I \to J$ in $\BCat$ the coproduct
  $\coprod_f \from \TCat_I \to \TCat_J$ along $f$ exists, then
  the coproduct $\coprod_{\chLift{f}} \from \TCCat_I \to \TCCat_J$
  along $\chLift{f}$ is given by $\chLift{\coprod_f}$.
  Similarly, the product $\prod_{\chLift{f}}$ along $\chLift{f}$ is given
  by $\chLift{\prod_f}$.
\end{theorem}
\begin{proof}
  By \iLemRef{fibres-are-chain-cats} and \iLemRef{indexed-lifting}, we obtain
  that an adjunction $\coprod \dashv f$ lifts to an adjunction
  $\chLift{\coprod_f} \dashv \chLift{f}$.
  Hence, the coproduct along $\chLift{f}$ is given by $\chLift{\coprod_f}$.
  \qedhere
\end{proof}

\begin{example}
  Both $\Pred$ and $\Fam{\SetC}$ are well known to have products and coproducts
  along any function in $\SetC$.
  We note that also $\qPred$ has products along all functions $f \from X \to Y$,
  given by
  \begin{equation*}
    \prodN_f(\delta \from X \to \I)(y)
    = \inf \setDef{\delta(x)}{x \in X, f(x) = y}.
  \end{equation*}
  Finally, in a syntactic logic, as in \iExRef{syntactic-logic-CCC},
  one has that $L \to \mathcal{C}$ obtains products and coproducts along
  projections $(\Gamma, x : A) \to \Gamma$ from universal and existential
  quantification over $A$, respectively.
  To have arbitrary (co)products, one additionally needs an equality relation
  in the logic, cf.~\cite{Jacobs1999-CLTT}.
  By \iThmRef{coproduct-lift}, all these products and coproducts lift to
  the corresponding fibration of descending chains.
  \qedDef
\end{example}

Let us denote for $I \in \BCat$ the later modality $\laterFib{K_I}$ on
$\TCCat_I$ by $\laterFib{I}$.
We can then establish the following essential properties about the interaction
of the later modalities and (co)products, which are analogue
to those in~\cite[cf.~Thm.~2.7]{Birkedal12:GuardedDomainTheory}.
This theorem allows one to distribute in proofs quantifiers over the later
modality.
\begin{theorem}
  The following holds for fibred products and coproducts in $\pC$.
  \begin{itemize}
  \item There is an isomorphism
    $\laterFib{J} \comp \prod_{\chLift{f}}
    \cong \prod_{\chLift{f}} \comp \laterFib{I}$.
  \item There is a natural transformation
    $\iota \from \coprod_{\chLift{f}} \comp \laterFib{I}
    \natTo \laterFib{J} \comp \coprod_{\chLift{f}}$.
    Moreover, if $f$ is inhabited, that is, has a section $g \from J \to I$,
    then $\iota$ has a section $\iota^g$.
  \end{itemize}
\end{theorem}
\begin{proof}
  Establishing the sought-after isomorphism and $\iota$ is straightforward.
  The section $\iota^g$ of $\iota$ for a given $g \from J \to I$ is can be
  defined by
  \begin{align*}
    \iota^g_{\sigma,0} & =
    \T_J \xrightarrow{\T_g}
    \T_I \xrightarrow{\eta}
    \reidx{f} \coprodN_f \T_I \xrightarrow{\cartL{f}{\coprodN_f \T_I}}
    \coprodN_f \T_I
    \\
    \iota^g_{\sigma,n+1} & =
    \coprodN_f \sigma_n \xrightarrow{\id} \coprodN_f \sigma_n.
  \end{align*}
  That this is a right-inverse of $\iota$ follows from finality if $\T_J$.
  \qedhere
\end{proof}
\begin{longproof}
  \begin{itemize}
  \item We define
    \begin{equation*}
      \kappa \from
      \laterFib{Y} \comp \prod_{\chLift{f}}
      \cong \prod_{\chLift{f}} \comp \laterFib{X}
      \colon \kappa^{-1}
    \end{equation*}
    by
    \begin{align*}
      & \kappa_{\sigma,0} =
      \T_Y
      \xrightarrow{\eta} \prodN_f \reidx{f} \T_Y
      \xrightarrow{\prod_f !} \prodN_f \T_X
      & & \kappa_{\sigma,n+1} =
      \prodN_f \sigma_n \xrightarrow{\id} \prodN_f \sigma_n
    \end{align*}
    and
    \begin{align*}
      & \kappa^{-1}_{\sigma,0} = \prodN_f \T_X \xrightarrow{!} \T_Y
      & & \kappa^{-1}_{\sigma,n+1} =
      \prodN_f \sigma_n \xrightarrow{\id} \prodN_f \sigma_n
    \end{align*}
    \begin{itemize}
    \item Naturality of $\kappa^{-1}$ is easily checked.
    \item $\kappa^{-1} \comp \kappa = \id$ is trivial because $\T_Y$ is final
    \item For $\kappa \comp \kappa^{-1} = \id$, we check at $0$:
      \begin{align*}
        \kappa_{\sigma,0} \comp \kappa^{-1}_{\sigma,0}
        & = \prodN_f! \comp \eta \comp ! \\
        & = \prodN_f! \comp \prodN_f \reidx{f} ! \comp \eta
        \tag*{$\eta$ natural} \\
        & = \prodN_f (! \comp \reidx{f} !) \comp \eta
        \tag*{$\prodN_f$ functor} \\
        & = \prodN_f (! \comp \varepsilon_{\T_X}) \comp \eta
        \tag*{$\T_X$ final} \\
        & = \prodN_f ! \comp \prodN_f \varepsilon_{\T_X} \comp \eta
        \tag*{$\prodN_f$ functor} \\
        & = \prodN_f ! \comp \id_{\prodN_f \T_X}
        \tag*{$\reidx{f} \dashv \prodN_f$} \\
        & = \prodN_f \id \comp \id_{\prodN_f \T_X}
        \tag*{$\T_X$ final} \\
        & = \prodN_f \id_{\T_X} \comp \id_{\prodN_f \T_X}
        \tag*{$\T_X$ final} \\
        & = \id_{\prodN_f \T_X}
        \tag*{$\prodN_f$ functor}
      \end{align*}
    \end{itemize}
    Thus, $\kappa^{-1}$ is natural and $\kappa$ its point-wise inverse.
    Hence, $\kappa$ is also natural, and they form together an isomorphism.
  \item We define $\iota$ by
    \begin{align*}
      & \iota_{\sigma,0} = \coprodN_f \T_X \xrightarrow{!} \T_Y
      & & \iota_{\sigma,n+1} =
      \coprodN_f \sigma_n \xrightarrow{\id} \coprodN_f \sigma_n,
    \end{align*}
    which is easily seen to be natural like $\kappa^{-1}$ above.
    We define a section $\iota^g$ by
    \begin{align*}
      \iota^g_{\sigma,0} & =
      \T_Y \xrightarrow{\T_g}
      \T_X \xrightarrow{\eta}
      \reidx{f} \coprodN_f \T_X \xrightarrow{\cartL{f}{\coprodN_f \T_X}}
      \coprodN_f \T_X
      \\
      \iota^g_{\sigma,n+1} & =
      \coprodN_f \sigma_n \xrightarrow{\id} \coprodN_f \sigma_n,
    \end{align*}
    which is easily seen to be natural.
    Finally, that $\iota^g$ is a section of $\iota$ follows immediately
    from finality of $\T_Y$.
    \qedhere
  \end{itemize}
\end{longproof}

\begin{remark}
  \label{rem:fibred-products-hard}
  It should be possible to establish in $\TCCat$ fibred products and coproducts
  along general morphisms of $\BCCat$.
  However, this is a much more difficult task, which will use ideas similar
  to those used in \iThmRef{fibred-diagrams-CCC}.
  Intuitively, the products that we established correspond to universal
  quantifiers over fixed sets, while general products would correspond to
  universal quantification over variable sets.
  The difference is analogous to that in Kripke models of (intuitionistic)
  first-order logic:
  Suppose $\mathcal{M} = (W, \leq, U)$ is a model, where $\leq$ is a partial
  order on $W$ and $U$ an interpretation for the quantification domain.
  If $U$ is merely a set, then the satisfaction $\vDash$ relation is defined
  for universal quantification by
  \begin{equation*}
    w, \rho \vDash \all{x} \varphi
    \quad \iff \quad
    \all{u \in U} w, \rho[x \mapsto u] \vDash \varphi.
  \end{equation*}
  However, if $U$ is a family $U \from W \to \SetC$, then the interpretation
  of universal quantification involves a quantification over all successor
  worlds:
  \begin{equation*}
    w, \rho \vDash \all{x} \varphi
    \quad \iff \quad
    \all{w \leq v} \all{u \in U(v)} v, \rho[x \mapsto u] \vDash \varphi.
  \end{equation*}
  This means that if we want to lift products to general chains, then
  the fibred products will involve again a quantification over morphisms
  in the index category, and the product must also be given by an end,
  as we used it in the construction of exponents in
  \iThmRef{fibred-diagrams-CCC}.
  Since this construction is fairly involved and not necessary for our
  current purposes, we will leave such a construction aside for now.
  \qedDef
\end{remark}

We finish this section by lifting also the construction of fixed points
for contractive maps to fibrations.
\begin{theorem}
  \label{thm:lob-fib}
  For every $\sigma \in \TCCat_c$ there is a unique map in $\TCCat_c$,
  dinatural in $\sigma$,
  \begin{equation*}
    \lobFib{c}_\sigma \from \sigma^{\later \sigma} \to \sigma,
  \end{equation*}
  such that for all $g$-contractive maps $f$ the
  map $\lobFib{c}_\sigma \comp \abstr{g}$ is a solution for $f$.
\end{theorem}
\begin{proof}
  We define $\lobFib{c}_n \from (\sigma^{\later \sigma})_n \to \sigma_n$ again
  by iteration on $n$.
  For $0$, we put
  \begin{equation*}
    \lob_0 \colonequals
    (\sigma^{\later \sigma})_0
    \xrightarrow{\pair{\id, !}} (\sigma^{\later \sigma})_0 \times \T_{c_0}
    = (\sigma^{\later \sigma})_0 \times (\later \sigma)_0
    \xrightarrow{\ev_0} \sigma_0,
  \end{equation*}
  where $\ev$ is the counit of
  $(-) \times \later \sigma \dashv (-)^{\later \sigma}$.
  In the iteration step, we first define a morphism $\step_n$ as the mediating
  morphism in the following diagram.
  \begin{equation*}
    \begin{tikzcd}
      \parens*{\sigma^{\laterFib{c} \sigma}}_{n+1}
      \rar{\parens*{\sigma^{\laterFib{c} \sigma}}(n \leq n+1)}
      \dar[dashed]{\step_n}
      & \parens*{\sigma^{\laterFib{c} \sigma}}_{n}
      \dar{\lobFib{c}_n} \\
      \parens*{\laterFib{c} \sigma}_{n+1} = \reidx{c(n \leq n+1)} \sigma_n
      \rar{\cartL{c(n \leq n+1)}{\sigma_n}}
      & \sigma_n
    \end{tikzcd}
    \xmapsto{p}
    \begin{tikzcd}
      c_{n+1} \dar[equal]{} \rar{c(n \leq n+1)}
      & c_n \dar[equal]{} \\
      c_{n+1} \rar{c(n \leq n+1)}
      & c_n
    \end{tikzcd}
  \end{equation*}
  The map $\lobFib{c}_{n+1}$ is then given by
  \begin{equation*}
    \lobFib{c}_{n+1} \colonequals
    (\sigma^{\later \sigma})_{n+1}
    \xrightarrow{\pair{\id, \step_n}}
    (\sigma^{\later \sigma})_{n+1} \times (\laterFib{c} \sigma)_{n+1}
    \xrightarrow{\ev_{n+1}} \sigma_{n+1},
  \end{equation*}
  That $\lobFib{c}$ is vertical, i.e., $\pC(\lobFib{c}) = \id$ is clear
  from the definition.
  The other properties follows like in \iThmRef{lob}.
  \qedhere
\end{proof}

%% file: content/fin-chain-up-to.tex
\section{The Final Chain and Up-To Techniques}
\label{sec:fin-chain-up-to}

Having laid the ground work, we come now to the actual objects of interest:
coinductive predicates.
We will proceed again in two steps, in that we first present coinductive
predicates over arbitrary categories and then move to fibrations.
The following captures the usual construction of the final chain.
\begin{definition}
  \label{def:fin-chain}
  Let $\Cat{C}$ be a category with a final object and
  $\Phi \from \Cat{C} \to \Cat{C}$ a functor.
  We define a chain $\ch{\Phi} \in \CC$ by
  \begin{equation*}
    \begin{aligned}
      \ch{\Phi}_0 & = \T \\
      \ch{\Phi}_{n+1} & = \Phi(\ch{\Phi}_n)
    \end{aligned}
    \; \text{ and } \;
    \ch{\Phi}(m \leq n) =
    \begin{cases}
      ! \from \ch{\Phi}_n \to \T, & m = 0 \text{ or } n = 0 \\
      \Phi(\ch{\Phi}(m' \leq n')), & m = m' + 1, n = n' + 1
    \end{cases}
  \end{equation*}
\end{definition}

The following theorem will play a central role in recursive proofs,
as it allows us to unfold $\ch{\Phi}$ and thereby to make progress in a
recursive proof.
Additionally, it tells us that $\ch{\Phi}$ is a fixed point of the
functor $\later \comp \, \chLift{\Phi}$,
cf.~\cite[Thm.~2.14]{Birkedal12:GuardedDomainTheory}.
\begin{theorem}
  \label{thm:step}
  We have that $\ch{\Phi} = \later \bparens{\chLift{\Phi} \ch{\Phi}}$.
\end{theorem}
\begin{longproof}
  For $n \in \N$ we have
  \begin{equation*}
    \ch{\Phi}_n
    =
    \begin{cases}
      \T_X, & n = 0 \\
      \Phi(\ch{\Phi}_k), & n = k+1
    \end{cases}
    = (\later (\chLift{\Phi} \ch{\Phi}))_n
  \end{equation*}
  and similarly
  $\ch{\Phi}(m \leq n) = (\later (\chLift{\Phi} \ch{\Phi}))(m \leq n)$.
  \qedhere
\end{longproof}

Just as important as unfolding $\ch{\Phi}$ is the ability to remove
contexts, use transitivity of relations etc. in a proof.
Such properties can properties can be captured through so-called
compatible up-to techniques~\cite{Bonchi:UpTo-Fib-LICS,Rot2014-EnhCoalgBisim}.
\begin{theorem}
  \label{thm:compat-proof-rules}
  Let $T$ and $\Phi$ be functors $\Cat{C} \to \Cat{C}$.
  If there is a natural transformation $\rho \from T \Phi \natTo \Phi T$,
  then there is a map
  $\ch{\rho} \from \chLift{T}\ch{\Phi} \to \ch{\Phi}$ in $\CC$.
\end{theorem}
\begin{proof}
  We define $\ch{\rho}_n$ by iteration on $n$:
  \begin{align*}
    \ch{\rho}_0 & = T\T \xrightarrow{!} \T \\
    \ch{\rho}_{n+1} & =
    \displaystyle
    T \Phi \ch{\Phi}_{n} \xrightarrow{\textstyle \rho_{\ch{\Phi}_n}}
    \Phi T \ch{\Phi}_{n} \xrightarrow{\textstyle \Phi \ch{\rho}_n}
    \Phi \ch{\Phi}_{n} = \ch{\Phi}_{n+1}
  \end{align*}
  That $\ch{\rho}$ is a morphism in $\ch{\Cat{C}}$ follows easily
  by induction, and by using naturality of $\rho$ functoriality of $\Phi$.
  \qedhere
\end{proof}
\begin{longproof}
  We define $\ch{\rho}_n$ by iteration on $n$:
  \begin{align*}
    \ch{\rho}_0 & = T\T \xrightarrow{!} \T \\
    \ch{\rho}_{n+1} & =
    \displaystyle
    T \Phi \ch{\Phi}_{n} \xrightarrow{\textstyle \rho_{\ch{\Phi}_n}}
    \Phi T \ch{\Phi}_{n} \xrightarrow{\textstyle \Phi \ch{\rho}_n}
    \Phi \ch{\Phi}_{n} = \ch{\Phi}_{n+1}
  \end{align*}
  To prove that $\ch{\rho}$ is a morphism in $\ch{\Cat{C}}$, that is, that
  for all $m \leq n$ we have
  $\ch{\rho}_m \comp T \ch{\Phi}(m \leq n)
  = \ch{\Phi}(m \leq n) \comp \ch{\rho}_n$,
  we proceed by induction
  on $n$ and $m$.
  The cases that $n$ or $m$ are $0$, are easy.
  For the induction step, we need that the following diagram commutes.
  \begin{equation*}
    \begin{tikzcd}[column sep=large]
      T\Phi\ch{\Phi}_n \rar{\textstyle \rho_{\ch{\Phi}_n}}
      \dar{\textstyle T\Phi\ch{\Phi}(m \leq n)}
      & \Phi T \ch{\Phi}_n \rar{\textstyle \Phi \ch{\rho}_n}
      \dar{\textstyle \Phi T \ch{\Phi}(m \leq n)}
      & \Phi\ch{\Phi}_n
      \dar{\textstyle \Phi\ch{\Phi}(m \leq n)}
      \\
      T\Phi\ch{\Phi}_m \rar{\textstyle \rho_{\ch{\Phi}_m}}
      & \Phi T \ch{\Phi}_m \rar{\textstyle \Phi \ch{\rho}_m}
      & \Phi\ch{\Phi}_m
    \end{tikzcd}
  \end{equation*}
  The left rectangle commutes by naturality of $\rho$, while the right triangle
  commutes by functoriality of $\Phi$ and the the induction hypothesis.
  \qedhere
\end{longproof}

\begin{remark}
  Pous and Rot~\cite{Pous17:CompanionCodensCausal} prove a result similar to
  \iThmRef{compat-proof-rules}, namely that a monotone function $T$ on a
  complete lattice is below the companion of $\Phi$ if and only if there is a
  map $\chLift{T}\ch{\Phi} \to \ch{\Phi}$.
  This result is equivalent to \iThmRef{compat-proof-rules} because the
  companion itself is compatible.
  \qedDef
\end{remark}

Let $\Delta_n \from \Cat{C} \to \Cat{C}^n$ be the diagonal functor and put
$\Phi^{\times n} =
\underbrace{\Phi \times \dotsm \times \Phi}_{n \text{-times}}$.
We then obtain the following corollary of \iThmRef{compat-proof-rules}, which
allows us its application to compatible up-to techniques that have
$n$ arguments.
For example, the transitive closure of a relation requires $2$ arguments,
see~\cite{Bonchi:UpTo-Fib-LICS} for details.
\begin{corollary}
  Let $n \in \N$ and $T \from \Cat{C}^n \to \Cat{C}$ be a functor.
  If there is a natural transformation
  $\rho \from T \Phi^{\times n} \natTo \Phi T$, then there is a map
  $\ch{\rho} \from \chLift{T} (\chLift{\Delta_n} \ch{\Phi}) \to \ch{\Phi}$ in
  $\ch{\Cat{C}}$.
\end{corollary}
\begin{longproof}
  Let $T' = T \comp \Delta_n \from \Cat{C} \to \Cat{C}$.
  Note now that $\Delta_n\Phi = \Phi^{\times n} \Delta_n$, thus we have
  \begin{equation*}
    T'\Phi = T \Phi^{\times n} \Delta_n
    \xRightarrow{\rho \Delta_n} \Phi T \Delta_n
    = \Phi T'.
  \end{equation*}
  From \iThmRef{compat-proof-rules}, we obtain a morphism
  $\ch{\rho \Delta_n} \from \chLift{T'}\ch{\Phi} \to \ch{\Phi}$.
  As functor liftings to chains compose (cf.~\iLemRef{seq-fibred-functor}),
  we can define $\ch{\rho} = \ch{\rho \Delta_n}$.
%  $\from \chLift{T}(\chLift{\Delta_n} \ch{\Phi}) \to \ch{\Phi}$.
  \qedhere
\end{longproof}

Let us now move to the setting of fibrations.
For the remainder of this section, we assume to be given a functor
$F \from \BCat \to \BCat$ that describes the behaviour of coalgebras,
and a lifting $G \from \TCat \to \TCat$ of $F$ that describes a predicate
on $F$-coalgebras, see~\cite{Hasuo13:CoindPredFinalSeqFib} for a more detailed
introduction.
\begin{assumption}
  We assume to be given a map of fibrations $(F, G) \from p \to p$
  and a coalgebra $c \from X \to FX$ in $\BCat$.
  Moreover, we require that $\BCat$ has a final object.
\end{assumption}

Under these assumption, we can define a functor $\Phi \from \TCat_X \to \TCat_X$
by
\begin{equation*}
  \Phi \colonequals \reidx{c} \comp \, G \from \TCat_X \to \TCat_X,
\end{equation*}
which describes, what is often called, a \emph{predicate transformer}.
A coalgebra for $\Phi$ is then referred to as a \emph{$\Phi$-invariant}.
One can now talk about up-to techniques for $G$ and for $\Phi$.
Both kinds are related by the following result, which allows us to obtain
compatible up-to techniques on fibres from global ones.
\begin{theorem}
  Let $T \from \TCat \to \TCat$  be a a lifting of the identity $\Id_{\TCat}$.
  If there is natural transformation $\rho \from TG \natTo GT$ with
  $P\rho = \id F \from F \natTo F$, then there is a natural transformation
  $\rho^c \from T\Phi \natTo \Phi T$ with $P\rho^c = \id \from \Id \natTo \Id$.
\end{theorem}
\begin{longproof}
  We define $\rho^c_X \from T\Phi X \to \Phi TX$ as the mediating morphism in
  the following diagram.
  \begin{equation*}
    \begin{tikzcd}[column sep=large]
      T\reidx{c}GX \rar{T\cartL{c}GX} \dar[dashed]{\rho^c_X}
      & TGX \dar{\rho_X} \\
      \reidx{c}GTX \rar{\cartL{c}GTX}
      & GTX
    \end{tikzcd}
    \overset{P}{\mapsto}
    \begin{tikzcd}[column sep=large]
      X \rar{c} \dar{\id_X}
      & FX \dar{\id_{FX}} \\
      X \rar{c}
      & FX
    \end{tikzcd}
  \end{equation*}
  That $\rho^c$ is natural follows by uniqueness of mediating morphisms.
  \qedhere
\end{longproof}

Similarly, one obtains also a descending chain for $G$.
\begin{lemma}
  Let $(F,G)$ be a lifting to $p \from \TCat \to \BCat$.
  Then $\ch{G} \in \TCCat_{\ch{F}}$.
\end{lemma}

The global chain $\ch{G}$ is again related to the local one $\ch{\Phi}$ as
follows.
From the coalgebra $c \from X \to FX$, we define a morphism
$\ch{c} \from K_X \to \ch{F}$ in $\TCCat$ iteratively by
\begin{equation*}
  \ch{c}_0 = \, !_X \from X \to \T
  \quad \text{ and } \quad
  \ch{c}_{n+1} = X \xrightarrow{c} FX \xrightarrow{F \ch{c}_n} \ch{F}_n.
\end{equation*}

Using $\ch{c}$, we can relate the global and local chains.
\begin{proposition}
  \label{prop:rel-lifting-transformer}
  In $\TCCat_X$, we can find isomorphisms
  \begin{itemize}
  \item  $\reidx{\ch{c}} \bparens{\ch{G}} \cong \ch{\Phi}$ and
  \item $\reidx{(\nextOp \comp \ch{c})} \bparens{\later \ch{G}}
    \cong \laterFib{X} \ch{\Phi}$.
  \end{itemize}
\end{proposition}
\begin{longproof}
  \begin{itemize}
  \item We prove $\bparens{\reidx{\ch{c}} \bparens{\ch{G}}}_n \cong \ch{\Phi}_n$
    by induction on $n$.
    For $0$, we have
    \begin{align*}
      \bparens{\reidx{\ch{c}} \bparens{\ch{G}}}_0
      & = \reidx{\ch{c}_0} \ch{G}_0 \\
      & = \reidx{!_X} (\T) \\
      & \cong \T_X
      \tag*{Fibred fin. objects} \\
      & = \ch{\Phi}_0
    \end{align*}
    For the induction step, we have
    \begin{align*}
      \bparens{\reidx{\ch{c}} \bparens{\ch{G}}}_{n+1}
      & = \reidx{\ch{c}_{n+1}} \ch{G}_{n+1} \\
      & = \reidx{\bparens{F\ch{c}_n \comp c}}(G(\ch{G}_n)) \\
      & \cong \reidx{c} \bparens{\reidx{\bparens{F\ch{c}_n}}(G(\ch{G}_n))}
      \tag*{Cloven fibration} \\
      & \cong \reidx{c} \bparens{G(\reidx{\ch{c}_n}\ch{G}_n)}
      \tag*{Lifting} \\
      & \cong \reidx{c} \bparens{G(\ch{\Phi}_n)}
      \tag*{Induction Hypothesis} \\
      & = \ch{\Phi}_{n+1}
    \end{align*}
  \item Second part:
    \begin{align*}
      & \bparens{\reidx{(\nextOp \comp \ch{c})} \bparens{\later \ch{G}}}_n \\
      & = \reidx{(\nextOp \comp \ch{c})_n} \bparens{\later \ch{G}}_n
      \tag*{Def. of reindexing}
      \\
      & =
      \begin{cases}
        \reidx{!_X}(\T), & n = 0 \\
        \reidx{\bparens{F^k! \comp \ch{c}_{k+1}}}(\ch{G}_k), & n = k+1
      \end{cases}
      \tag*{Def. of $\nextOp$}
      \\
      & \cong
      \begin{cases}
        \T_X, & n = 0 \\
        \reidx{\ch{c}_{k}}(\ch{G}_k), & n = k+1
      \end{cases}
      \tag*{Fibred final objects and $\ch{c}$ map of chains}
      \\
      & = \bparens{\laterFib{?} \bparens{\reidx{\ch{c}} \bparens{\ch{G}}}}_n
      \tag*{Def. $\laterFib{?}$}
      \\
      & \cong \bparens{\laterFib{?} \ch{\Phi}}_n
      \tag*{By first part}
    \end{align*}
    \qedhere
  \end{itemize}
\end{longproof}

From \iPropRef{rel-lifting-transformer}, we can obtain an alternative
proof of one of the central results (Thm.~3.7.i) by
Hasuo et al.~\cite{Hasuo13:CoindPredFinalSeqFib}.
\begin{corollary}
  We have
  $\lim \ch{\Phi} \cong \reidx{c_{\omega}}(\lim \ch{G})$,
  where $c_{\omega} \from X \to \lim \ch{F}$ is the unique map induced by
  $\ch{c}$ and the limit property.
\end{corollary}
\begin{proof}
  We have
  \begin{align*}
    \SwapAboveDisplaySkip
    \lim \ch{\Phi}
    & \cong \lim (\reidx{\ch{c}}\ch{G})
    \tag{By \iPropRef{rel-lifting-transformer}} \\
    & \cong \lim (\reidx{\ch{\pi \comp c_\omega}}\ch{G})
    \tag{Def. of $c_\omega$} \\
    & \cong \lim (\reidx{\ch{c_\omega}} (\reidx{\ch{\pi}} \ch{G}))
    \tag{Cloven fibration} \\
    & \cong \reidx{c_\omega} (\lim (\reidx{\ch{\pi}} \ch{G}))
    \tag*{Fibred limits} \\
    & \cong \reidx{c_\omega} (\lim \ch{G})
    \tag{\cite[Lem. 3.5]{Hasuo13:CoindPredFinalSeqFib}}
  \end{align*}
  \qedhere
\end{proof}

If the chain $\ch{\Phi}$ converges in $\omega$ steps, then we obtain
soundness and completeness for proofs given over $\ch{\Phi}$.
This result is a trivial reformulation of the usual construction of
final coalgebra.
However, the present formulation is more convenient in the context of the
the recursive proofs that we construct by appealing to the later modality,
as those will be maps in $\TCCat_X$.
\begin{proposition}
  Suppose $\nu \Phi$ is a coinductive predicate, that is, there is a
  final coalgebra  $\xi \from \nu \Phi \to \Phi(\nu \Phi)$.
  If $\Phi$ preserves $\op{\omega}$-limits, then
  % the inverse of the canonical map $\Phi(\lim \ch{\Phi}) \to \lim \ch{\Phi}$
  % gives rise to a final coalgebra $\xi \from \nu \Phi \to \Phi(\nu \Phi)$
  % of $\Phi$.
  % In this case,
  maps $A \to \nu \Phi$ in $\TCat_X$ are given equivalently by
  maps $K_A \to \ch{\Phi}$ in $\TCCat_X$.
\end{proposition}

%% file: content/examples.tex
\section{Examples}
\label{sec:examples}

In this last section, we demonstrate how the framework that we developed can
be used to obtain recursive proofs for coinductive predicates over different
kinds of first-order logic.
The first example is thereby in the setting of set-based predicates.
\begin{example}
  In this example, we define a predicate on streams that expresses that a
  real-valued stream is greater than $0$ everywhere and use the developed
  framework to prove that a certain stream is in the predicate.
  This example is fairly straightforward, but still has all the ingredients
  to illustrate the framework.

  Let $F \from \SetC \to \SetC$ and $G \from \Pred \to \Pred$ be given by
  $F = \R \times \Id $ and
  $G(X, P) = (F X, \setDef{(a, x)}{a > 0 \conj x \in P})$.
  It is easy to show that $G$ is a lifting of $F$, and we obtain the predicate
  of streams that are larger than $0$ everywhere as the final coalgebra of
  the functor $\Phi \from \Pred_{\Str{\R}} \to \Pred_{\Str{\R}}$ with
  $\Phi = \reidx{\pair{\head, \tail}} \comp G$.

  Next, we define for $a \in \R$ the constant stream $\constStr{a}$ by
  the following stream differential equation (SDE)~\cite{HKR14:SDE}.
  \begin{equation*}
    \constStr{a}_0 = a
    \qquad (\constStr{a})' = \constStr{a}
  \end{equation*}
  Similarly, we can define the point-wise addition of streams by
  \begin{equation*}
    (s \oplus t)_0 = s_0 + t_0
    \qquad (s \oplus t)' = s' \oplus t'.
  \end{equation*}
  Finally, let $s \in \Str{\R}$ be given by the following SDE.
  \begin{equation*}
    s_0 = 1 \qquad s' = \constStr{1} \oplus s.
  \end{equation*}
  Our goal is to prove that $s$ is greater than $0$ everywhere, that is,
  we want to prove that $s$ is in the final coalgebra $\nu \Phi$ of the above
  $\Phi$.
  Since the tail $s'$ of $s$ defined of $\constStr{1} \oplus -$, the following
  up-to technique will be handy.
  Let us define $C \from \chCat{\Pred}_{\Str{\R}} \to \chCat{\Pred}_{\Str{\R}}$
  to be
  \begin{equation*}
    C(P) = \setDef{\constStr{1} \oplus t}{t \in P}.
  \end{equation*}
  One easily shows that $C$ is $\Phi$-compatible, that is,
  $C\Phi \subseteq \Phi C$.
  In fact, this follows from point-wise addition being causal,
  see~\cite{Rot15:Thesis,Pous17:CompanionCodensCausal}.
  Thus, we have by \iThmRef{compat-proof-rules} that
  $\chLift{C} \ch{\Phi} \sqsubseteq \ch{\Phi}$,
  where $\sqsubseteq$ is the point-wise inclusion of indexed predicates.

  Given an indexed predicate $\sigma \in \chCat{\Pred}_X$, we define
  \begin{equation*}
    \vdash \sigma \colonequals \chCat{\T}_X \sqsubseteq \sigma.
  \end{equation*}
  Hence, $\vdash \sigma$ holds if there is a morphism
  $\chCat{\T}_X \to \sigma$ in $\chCat{\Pred}_X$.
  Given $x \in X$, we define the predicate
  $x \inC \sigma$ in $\chCat{\Pred}_X$ to be the following exponential
  in $\chCat{\Pred}_X$.
  \begin{equation*}
    x \inC \sigma \colonequals \sigma^{K_{\set{x}}}.
  \end{equation*}
  Spelling out these definitions, one easily finds that
  \begin{equation*}
    \vdash x \inC \sigma
    \iff \all{n \in \N} x \in \sigma_n.
  \end{equation*}

  For brevity, let us write $\varphi \colonequals s \inC \ch{\Phi}$
  and $\later$ for $\laterFib{\Str{\R}}$.
  Using the previous results, we now obtain a proof for
  $\vdash \varphi$ as follows, where each proof step is given applying
  the indicated construction in $\chCat{\Pred}_{\Str{\R}}$.
  \begin{equation*}
    \AxiomC{}
    \UnaryInfC{$s_0 > 0$}
    \AxiomC{}
    \RightLabel{(Identity)}
    \UnaryInfC{$\later \varphi \vdash
      \later \parens*{s \inC \ch{\Phi}}$}
    \RightLabel{(Def. $C$)}
    \UnaryInfC{$\later \varphi \vdash
      \later \parens*{\constStr{1} \oplus s \inC \chLift{C}(\ch{\Phi})}$}
    \RightLabel{($C$ compatible)}
    \UnaryInfC{$\later \varphi \vdash
      \later \parens*{\constStr{1} \oplus s \inC \ch{\Phi}}$}
    \RightLabel{(Def. of $s$)}
    \UnaryInfC{$\later \varphi \vdash
      \later \parens*{s' \inC \ch{\Phi}}$}
    \RightLabel{($\later$ pres. products)}
    \BinaryInfC{$\later \varphi \vdash
      \later \parens*{s \inC \chLift{\Phi}(\ch{\Phi})}$}
    \RightLabel{($\later$ functor)}
    \UnaryInfC{$\later \varphi \vdash
      s \inC \later \parens*{\chLift{\Phi}(\ch{\Phi})}$}
    \RightLabel{(Step, \iThmRef{step})}
    \UnaryInfC{$\later \varphi \vdash s \inC \ch{\Phi}$}
    \RightLabel{(Löb)}
    \UnaryInfC{$\vdash \varphi$}
    \DisplayProof
  \end{equation*}
  Thus, we have obtained a proof that $s$ is greater than $0$ everywhere
  purely by applying the category theoretical constructions presented
  in this paper.
  \qedDef
\end{example}

% The last example that we discuss is motivated by dependent type theories.
% It is, in fact, also an example, where the fibres are for once not merely
% pre-orders but rather proper categories.
% \begin{example}[Set Families]
  
% \end{example}

The next example shows that the same category theoretical setup that we
used to prove something above, can also be used to define functions.
\begin{example}
  Given a set $A$, we define a functor $F$ and a lifting $G^A$
  to the family fibration $\FamSet \to \SetC$ as follows.
  \begin{align*}
    & F \from \SetC \to \SetC
    & & G^A \from \FamSet \to \FamSet \\
    & F = \T + \Id
    & & G^A(I, X)_{u \in \T + I} =
    \begin{cases}
      \T, & u = \inj_1 \> \unit \\
      A \times X_v, & u = \inj_2 \> v
    \end{cases}
  \end{align*}
  $F$ has as final coalgebra the predecessor function
  $\pred \from \EN \to \T + \EN$ on the natural numbers extended with
  one element that indicates infinity.
  The family of so-called partial streams $\PStr_A$~\cite{Basold17:Phd} is the
  final coalgebra of $\Phi^A = \reidx{\pred} \comp G^A$.
  Our goal is now to define for a given $f \from A \to B$ a map
  $\PStr(f) \from \PStr_A \to \PStr_B$.
  Unfortunately, the results in~\cite{Hasuo13:CoindPredFinalSeqFib} do not
  apply here.
  But one can still show that $\Phi^A$ preserves $\op{\omega}$-limits,
  hence maps into $\PStr_A$ are equivalently given by maps into the chain
  $\ch{\Phi^A}$.
  Hence, we can obtain $\PStr(f)$ equivalently as a map
  $\ch{\Phi^A} \to \ch{\Phi^A}$ in $\chCat{\Fam{\SetC}}_{\EN}$.
  Denoting by $\Rightarrow$ the exponential in this fibre,
  we can construct the desired map by applying the following ``proof'' steps,
  where we write $u \mid \pred u = t \vdash X \to Y$ if we construct
  a map in $\Fam{\SetC}_{\EN}$ with the constraint that the index $u \in \EN$
  fulfils $\pred(u) = t$.
  \begin{equation*}
    \AxiomC{}
    \UnaryInfC{$u \mid \pred \> u = \inj_1 \> \unit \vdash
      \later ! \from (\later \ch{\Phi^A} \Rightarrow \later \ch{\Phi^B})_u
      \times \ch{\Phi^A}_u
      \to \later \T$}
    \AxiomC{$(\ast)$}
    \RightLabel{Cases for $G$}
    \BinaryInfC{$u \vdash (\later \ch{\Phi^A} \Rightarrow \later \ch{\Phi^B})_u
      \times \ch{\Phi^A}_u
      \to \later (\chLift{\Phi} \ch{\Phi^B})_u$}
    \RightLabel{Index abstraction}
    \UnaryInfC{$
      (\later \ch{\Phi^A} \Rightarrow \later \ch{\Phi^B}) \times \ch{\Phi^A}
      \to \later \chLift{\Phi} \ch{\Phi^B}$}
    \RightLabel{Step}
    \UnaryInfC{$
      (\later \ch{\Phi^A} \Rightarrow \later \ch{\Phi^B}) \times \ch{\Phi^A}
      \to \ch{\Phi^B}$}
    \RightLabel{Abstraction}
    \UnaryInfC{$(\later \ch{\Phi^A} \Rightarrow \later \ch{\Phi^B})
      \to (\ch{\Phi^A} \Rightarrow \ch{\Phi^B})$}
    \RightLabel{$\later$ functor}
    \UnaryInfC{$\later(\ch{\Phi^A} \Rightarrow \ch{\Phi^B})
      \to (\ch{\Phi^A} \Rightarrow \ch{\Phi^B})$}
    \RightLabel{Löb}
    \UnaryInfC{$\T \to (\ch{\Phi^A} \Rightarrow \ch{\Phi^B})$}
    \RightLabel{Uncurry}
    \UnaryInfC{$\ch{\Phi^A} \to \ch{\Phi^B}$}
    \DisplayProof
  \end{equation*}
  The step $(\ast)$ is thereby given as follows, where we write
  $S$ for $\later \ch{\Phi^A} \Rightarrow \later \ch{\Phi^B}$.
  \begin{equation*}
    \AxiomC{$\later (f \comp \proj_1 \comp \proj_2)
      \from S_u
      \times (\later A \times \later \ch{\Phi^A}_v)
      \to \later B$ \ \ and}
    \noLine
    \UnaryInfC{$\ev \comp (\id \times \proj_2) \from
      S_u
      \times (\later A \times \later \ch{\Phi^A}_v)
      \to \later \ch{\Phi^B}_v$}
    \RightLabel{Pairing}
    \UnaryInfC{$u \mid \pred \> u = \inj_2 \> v \vdash
      S_u
      \times (\later A \times \later \ch{\Phi^A}_v)
      \to \later B \times \later \ch{\Phi^B}_v$}
    \RightLabel{$\later$ pres. $\times$}
    \UnaryInfC{$u \mid \pred \> u = \inj_2 \> v \vdash
      S_u
      \times \later (A \times \ch{\Phi^A}_v)
      \to \later (B \times \ch{\Phi^B}_v)$}
    \RightLabel{Unfold}
    \UnaryInfC{$u \mid \pred \> u = \inj_2 \> v \vdash
      S_u
      \times \ch{\Phi^A}_u
      \to \later (B \times \ch{\Phi^B}_v)$}
    \DisplayProof
  \end{equation*}
\end{example}

\begin{longproof}
We can also instantiate the framework in a quantitative setting.
\begin{example}
  We can now reason about coinductive quantitative predicates.
  For instance, we can use $F = \I \times \Id$, giving us streams
  over the interval $\I$.
  As a lifting of $F$ to $\qPred$ we consider
  \begin{equation*}
    G(X, \delta) =
    \parens*{FX, \lam[(a, x)] \frac{1}{2} a + \frac{1}{2} \delta(x)}.
  \end{equation*}
  We put again $\Phi = \reidx{\pair{\head, \tail}} \comp G$.
  Suppose the final valuation $\delta \from X \to \I$ with
  $\delta = \Phi(\delta)$ exists.
  For a stream $s$ we then have
  \begin{align*}
    \delta(s)
    & = \frac{1}{2} s_0 + \frac{1}{2} \delta(s') \\
    & = \frac{1}{2} s_0 + \frac{1}{2} \parens*{
      \frac{1}{2} s_1 + \frac{1}{2} \delta(s'')} \\
    & = \frac{1}{2} s_0 + \frac{1}{4} s_1 + \frac{1}{4} \delta(s'') \\
    & = \dotsm \\
    & = \sum_{n \in \N} \frac{1}{2^{n+1}} s_n.
  \end{align*}
  Hence, $\delta$ calculates the sum of all entries in the stream $s_n$
  weighted by their position.
  Note that the sum is well-defined because of the exponential decay,
  hence
\end{example}
\end{longproof}

%% file: content/well-founded.tex
\section{General Well-Founded Orders}
\label{sec:well-founded}

Up to this point, we have used $\omega$ as fixed set with a well-founded on it.
As it turns out, it is not necessary to make this restriction and one can
construct the later modality and the Löb rule for any set with a well-founded
order on it.
This is similar to the development
in~\cite[Sec. 8]{Birkedal12:GuardedDomainTheory}.
The difference, however, is that Birkedal~et~al. require that the well-founded
set is a complete Heyting algebra and internalise the predecessor in there.
We will, in contrast, use properties of the category $\Cat{C}$, in which we
construct the sequences.
This approach is more in line with the previous development.

\begin{assumption}
  We assume that $(I, <)$ is a well-founded order and that $\Cat{C}$ has for
  each $\alpha \in I$ limits of the shape $\coslice{\alpha}{I}$.
\end{assumption}

Given these assumptions, we use now
\begin{equation*}
  \chCat{(-)} = \IntHom{I}{-}.
\end{equation*}

On $\CC$, we define
\begin{equation*}
  (\later \sigma)_{\alpha} = \lim_{\beta < \alpha} \sigma_{\beta}
\end{equation*}
with
$\proj^{\alpha}_\beta \from
\lim_{\beta < \alpha} \sigma_{\beta} \to \sigma_{\beta}$.
Since for $\alpha' \leq \alpha$ and $\beta' \leq \beta < \alpha' \leq \alpha$
we have
$\sigma(\beta' \leq \beta) \comp \proj^{\alpha}_\beta = \proj^{\alpha}_{\beta'}$,
we obtain a unique morphism
\begin{equation*}
  (\later \sigma)_{\alpha} = \lim_{\beta < \alpha} \sigma_{\beta}
  \xrightarrow{(\later \sigma)(\alpha' \leq \alpha)}
  (\later \sigma)_{\alpha'} = \lim_{\beta < \alpha'} \sigma_{\beta}.
\end{equation*}

\begin{theorem}
  For every $\sigma \in \CC$ there is a unique map in $\CC$,
  dinatural in $\sigma$,
  \begin{equation*}
    \lob_\sigma \from \sigma^{\later \sigma} \to \sigma
  \end{equation*}
  such that for all $g$-contractive maps $f$ the
  map $\lobFib{c}_\sigma \comp \abstr{g}$ is a solution for $f$.
\end{theorem}
\begin{proof}
  We construct $\lob_{\sigma, \alpha}$ by well-founded induction on $\alpha$.
  Thus, assume for all $\beta < \alpha$ that
  $\lob_{\sigma, \beta} \from
  \parens*{\sigma^{\later \sigma}}_{\beta} \to \sigma_\beta$
  exists and fulfils for all $\beta' \leq \beta < \alpha$
  \begin{equation*}
    \sigma(\beta' \leq \beta)
    \comp \lob_{\sigma,\beta}
    =
    \lob_{\sigma,\beta'}
    \comp \parens*{\sigma^{\later \sigma}}(\beta' \leq \beta).
  \end{equation*}
  By functoriality of $\sigma^{\later \sigma}$, we thus obtain
  \begin{equation*}
    \sigma(\beta' \leq \beta)
    \comp \lob_{\sigma,\beta}
    \comp \parens*{\sigma^{\later \sigma}}(\beta' \leq \alpha)
    =
    \lob_{\sigma,\beta'}
    \comp \parens*{\sigma^{\later \sigma}}(\beta \leq \alpha).
  \end{equation*}
  This gives us a unique morphism
  \begin{equation*}
    \parens*{\sigma^{\later \sigma}}_{\alpha}
    \xrightarrow{\step_{\alpha}}
    (\later \sigma)_{\alpha}
  \end{equation*}
  by the limit property.
  This allows us to define
  \begin{equation*}
    \lob_{\alpha} =
    \parens*{\sigma^{\later \sigma}}_{\alpha}
    \xrightarrow{\pair{\id, \step_{\alpha}}}
    \parens*{\sigma^{\later \sigma}}_{\alpha} \times (\later \sigma)_{\alpha}
    \xrightarrow{\ev_{\alpha}}
    \sigma_{\alpha},
  \end{equation*}
  which fulfils for all $\beta \leq \alpha$ that
  \begin{equation*}
    \sigma(\beta \leq \alpha)
    \comp \lob_{\sigma,\alpha}
    =
    \lob_{\sigma,\beta}
    \comp \parens*{\sigma^{\later \sigma}}(\beta \leq \alpha)
  \end{equation*}
  because of naturality of $\ev_\alpha$ in $\alpha$ and
  \begin{equation*}
    (\later \sigma)(\beta \leq \alpha) \comp \step_\alpha
    = \step_\beta \comp \parens*{\sigma^{\later \sigma}}(\beta \leq \alpha).
  \end{equation*}
  This latter equation follows easily from the limit property and
  the definition of $\step$.
  % if $\beta < \alpha$ ($\beta = \alpha$ is trivial), then
  % \begin{align*}
  %   \sigma(\beta \leq \alpha)
  %   \comp \lob_{\sigma,\alpha}
  %   & =
  %   \sigma(\beta \leq \alpha)
  %   \comp \ev_{\alpha} \comp \pair{\id, \step_\alpha} \\
  %   & =
  %   \ev_{\beta}
  %   \comp \parens*{\parens*{\sigma^{\later \sigma}}(\beta \leq \alpha)
  %     \times \sigma(\beta \leq \alpha)}
  %   \comp \pair{\id, \step_\alpha} \\
  %   & =
  %   \ev_{\beta}
  %   \comp \pair*{\parens*{\sigma^{\later \sigma}}(\beta \leq \alpha)
  %     , (\later \sigma)(\beta \leq \alpha) \comp \step_\alpha} \\
  %   & = \dots
  %   \tag{Use limit property here} \\
  %   & = \lob_{\sigma,\beta}
  %   \comp \parens*{\sigma^{\later \sigma}}(\beta \leq \alpha)
  % \end{align*}
  Similarly, one also proves by the limit property that $\lob$ is the
  unique dinatural transformation that allows the construction of solutions.
\end{proof}

%% file: content/concl.tex
\section{Conclusion and Future Work}
\label{sec:concl}

In this paper, we have established a framework that allows us to reason about
coinductive predicates in many cases by using recursive proofs.
At the heart of this approach sits the so-called later modality, which was
comes from provability logic~\cite{Beklemishev1999:ParameterFreeInduction,%
  Smorynski1985:SelfReference,Solovay1976:ProvabilityModalLogic}
but was later used to obtain guarded recursion in type
theories~\cite{Appel07:ModalTypeSystem,Atkey13:GuardedRec,Bizjak16:GuardedDTT,%
  Nakano00:ModalityRec}
and in domain theory~\cite{Birkedal:GuardedRecUniverseFP,%
  Birkedal12:GuardedDomainTheory}.
This modality allows us to control the recursion steps in a proof without
having to invoke parity or similar conditions~\cite{%
  Brotherston07:CompleteSequentCalcIndInfDecent,Fortier2013-CutsCirc,%
  Santocanale02:CircProofs,Simpson17:CyclicArithmeticPeano},
as we have seen in the examples in \iSecRef{examples}.
Moreover, even though similar
Birkedal~et~al.~\cite{Birkedal12:GuardedDomainTheory} obtained similar results,
their framework is limited to $\SetC$-valued presheaves, while our results are
applicable in a much wider range of situations, see the examples in
\iSecRef{fib-seq}.

So what is there left to do?
For once, we have not touched upon how to automatically extract a syntactic
logic and models from the fibration $\chCat{L} \to \chCat{\mathcal{C}}$ obtained
in \iExRef{syntactic-logic-CCC}.
This would subsume and simplify much of the development in~\cite{Basold17:Phd}.
Next, we discussed already in \iRemRef{fibred-products-hard} that the
construction of fibre products for general morphisms in fibrations of descending
chains is fairly involved.
However, such a construction would be useful, for example, to obtain
Kripke models abstractly.
Finally, also a closer analysis of the relation to proof systems obtained through
parameterised coinduction, the companion or cyclic proof systems would be
interesting.